\documentclass[12pt]{article}

\usepackage{amsmath}
\usepackage{amssymb}
\usepackage{amsfonts}
\usepackage{latexsym}
\usepackage{color}

\catcode `\@=11 \@addtoreset{equation}{section}

\catcode `\@=12

\newcommand{\be}{\begin{equation}}
\newcommand{\en}{\end{equation}}
\newcommand{\bea}{\begin{eqnarray}}
\newcommand{\ena}{\end{eqnarray}}
\newcommand{\beano}{\begin{eqnarray*}}
\newcommand{\enano}{\end{eqnarray*}}
\newcommand{\bee}{\begin{enumerate}}
\newcommand{\ene}{\end{enumerate}}

\newcommand{\mc}{\mathcal}

\newcommand{\D}{{\mc D}}

\newcommand{\E}{{\cal E}}
\newcommand{\F}{{\cal F}}
\newcommand{\G}{{\cal G}}

\newcommand{\Lc}{{\cal L}}
\newcommand{\ltwo}{{\Lc^2(\mathbb{R})}}

\newcommand{\1}{1 \!\! 1}

\newcommand{\Hil}{\mc H}

\newtheorem{thm}{Theorem}

\newtheorem{prop}[thm]{Proposition}
\newtheorem{defn}[thm]{Definition}

\newenvironment{proof}{\noindent {\bf Proof --}}{\hfill$\square$ \vspace{3mm}\endtrivlist}

\catcode `\@=11 \@addtoreset{equation}{section}
\catcode `\@=12

\begin{document}

\thispagestyle{empty}

\vspace*{2cm}

\begin{center}
{\Large \bf A class of weak pseudo-bosons and their bi-coherent states}   \vspace{2cm}\\

{\large F. Bagarello}\\
  Dipartimento di Ingegneria,
Universit\`a di Palermo,\\ I-90128  Palermo, Italy\\
and I.N.F.N., Sezione di Catania\\
e-mail: fabio.bagarello@unipa.it\\
home page: www1.unipa.it/fabio.bagarello

\end{center}

\vspace*{2cm}

\begin{abstract}
\noindent In this paper we extend some previous results on weak pseudo-bosons and on their related bi-coherent states. The role of {\em compatible} functions is discussed in details, and some examples are considered. The pseudo-bosonic ladder operators analysed in this paper generalize significantly those considered so far, and a class of new diagonalizable manifestly non self-adjoint Hamiltonians are deduced.

\end{abstract}

\vspace{2cm}

%{\bf PACS Numbers}:  .......

\vfill

%\pagenumbering{roman}

\newpage

\section{Introduction}

In quantum mechanics one of the preliminary difficulties one often meets, when dealing with a specific physical system, is to deduce the eigenvalues and the eigenvectors of its Hamiltonian. This is usually a difficult task. There exist very few systems for which this operation is simple, and only few for which it is not particularly complicated. Many more are the Hamiltonians whose eigenvectors and eigenvalues are difficult, when not impossible, to find. For this reason the quest for new {\em solvable} Hamiltonians is always open, and it has produced several interesting approaches: supersymmetric quantum mechanics, \cite{CKS,jun,gango,bagchi}, intertwining operators \cite{intop1,intop2,intop3}, factorizable Hamiltonians  and shape invariant potentials \cite{bagchi,dong}, ladder operators and generalized algebras \cite{curado2}, are just few of the approaches proposed along the years for giving partial results to this quest. And these techniques  have been adopted mainly in connection with self-adjoint Hamiltonians.

In recent years, the role of non self-adjoint Hamiltonians in physics has become more and more evident, and the interest for this kind of operators originated many lines of research, both with a physical and with a more mathematical taste. Thousands of paper have been published in the past two decades, together with some monographs and edited volumes. In particular, we refer to \cite{benbook}-\cite{bagbookPT}, where many more references can be found. 

The intersection between the above two topics has driven our interest to a specific way to solve the eigenvalue problem for a class of non self-adjoint Hamiltonians, written in terms of some sort of deformed bosonic operators. These have been called {\em pseudo-bosonic} operators, and we have shown in recent years that, under some mild assumptions, they produce two families of biorthonormal vectors which turn out to be eigenvectors of certain number-like, manifestly non self-adjoint, operators, $N$ and $N^\dagger$.  As one usually does in quantum mechanics, the analysis of these operators was originally performed in some relevant Hilbert space which is, quite often,  $\ltwo$. In this case, in most of the applications discussed in our knowledge the literature, the eigenfunctions of $N$ and $N^\dagger$ are total in $\ltwo$, but not  bases.  

Since 2020 it appeared clear that $\ltwo$ is not necessarily the most convenient space where to work with pseudo-bosons. In fact, \cite{bag2020JPA}, sometimes distributions are needed in this analysis. This opens the  possibility that  there exists some {\em intermediate} space between $\ltwo$ and $\D'(\mathbb{R})$, the set of distributions, which is relevant in the analysis of some deformed version of the canonical commutation relations. Indeed, we have already shown in \cite{bagJPCS2021} that, sometimes, more than square-integrable functions, it is convenient to work with pairs of {\em compatible} functions, i.e. with functions which are not both square integrable, but whose product still belongs to $\Lc^1(\mathbb{R})$. This is interesting, since allows us to introduce a concept of biorthogonality between functions, even outside  $\Lc^2(\mathbb{R})$. Moreover, this is in line with what has been done, with a more mathematical taste, in \cite{pip}, in connection with the so-called PIP spaces, where PIP stands for partial inner product. The importance of compatible functions will appear clear in the rest of this paper, where the focus is on pseudo-bosonic ladder operators which generalize many of those introduced in the past years. In doing so, we will find some interesting  mathematics, showing that a distributional point of view can be useful, in particular in connection with our version of coherent states.

 The paper is organized as follows: in the next section we will list few results and definitions on {\em ordinary} pseudo-bosons and bi-coherent states.  In Section \ref{sectgenclasspbs}, extending some older results, we consider a large class of pseudo-bosonic operators $a$ and $b$, and we construct two  families of functions, not necessarily in $\ltwo$, which are built by using $a$ and $b$ as ladder operators. Examples are described in Sections \ref{sectAGE} and \ref{sectexamples}, where we also deduce the Hamiltonians having these families of functions as eigenstates. Section \ref{sectBCS} contains our results on the weak version of bi-coherent states, while the conclusions are given in Section \ref{sectconcl}

\section{Preliminaries}

 To keep the paper self-contained, we devote this section to list few useful definitions and results on pseudo-bosons and on bi-coherent states  in Hilbert spaces. We refer to the recent book \cite{bagspringer} for many more details. We will show what happens {\em outside} Hilbert spaces in the second part of this paper.

\subsection{$\D$-pseudo bosons: basic facts}\label{sectpbs}

Let $\Hil$ be a given Hilbert space with scalar product $\left<.,.\right>$ and related norm $\|.\|$. Let $a$ and $b$ be two operators
on $\Hil$, with domains $D(a)\subset \Hil$ and $D(b)\subset \Hil$ respectively, $a^\dagger$ and $b^\dagger$ their adjoint, and let $\D$ be a dense subspace of $\Hil$
such that $a^\sharp\D\subseteq\D$ and $b^\sharp\D\subseteq\D$. Here with $x^\sharp$ we indicate $x$ or $x^\dagger$. Of course, $\D\subseteq D(a^\sharp)$
and $\D\subseteq D(b^\sharp)$.

\begin{defn}\label{def21}
	The operators $(a,b)$ are $\D$-pseudo bosonic  if, for all $f\in\D$, we have
	\be
	a\,b\,f-b\,a\,f=f.
	\label{A1}\en
\end{defn}

When $b=a^\dagger$, this is simply the canonical commutation relation (CCR) for ordinary bosons. However, when the CCR is replaced by (\ref{A1}), the situation changes. In particular, it is useful to assume the following:

\vspace{2mm}

{\bf Assumption $\D$-pb 1.--}  there exists a non-zero $\varphi_{ 0}\in\D$ such that $a\,\varphi_{ 0}=0$.

\vspace{1mm}

{\bf Assumption $\D$-pb 2.--}  there exists a non-zero $\Psi_{ 0}\in\D$ such that $b^\dagger\,\Psi_{ 0}=0$.

\vspace{2mm}
We have seen in \cite{bag2020JPA} that these assumptions are not necessarily true, even for operators satisfying (\ref{A1}). 

It is obvious that, since $\D$ is stable under the action of $b$ and $a^\dagger$, then  $\varphi_0\in D^\infty(b):=\cap_{k\geq0}D(b^k)$ and  $\Psi_0\in D^\infty(a^\dagger)$, so
that the vectors \be \varphi_n:=\frac{1}{\sqrt{n!}}\,b^n\varphi_0,\qquad \Psi_n:=\frac{1}{\sqrt{n!}}\,{a^\dagger}^n\Psi_0, \label{A2}\en
$n\geq0$, can be defined and they all belong to $\D$. Hence, they also belong to the domains of $a^\sharp$, $b^\sharp$ and $N^\sharp$, where $N=ba$.  Moreover, it is  simple to deduce the following lowering and raising relations:
\be
\left\{
\begin{array}{ll}
	b\,\varphi_n=\sqrt{n+1}\varphi_{n+1}, \qquad\qquad\quad\,\, n\geq 0,\\
	a\,\varphi_0=0,\quad a\varphi_n=\sqrt{n}\,\varphi_{n-1}, \qquad\,\, n\geq 1,\\
	a^\dagger\Psi_n=\sqrt{n+1}\Psi_{n+1}, \qquad\qquad\quad\, n\geq 0,\\
	b^\dagger\Psi_0=0,\quad b^\dagger\Psi_n=\sqrt{n}\,\Psi_{n-1}, \qquad n\geq 1,\\
\end{array}
\right.
\label{A3}\en as well as the eigenvalue equations $N\varphi_n=n\varphi_n$ and  $N^\dagger\Psi_n=n\Psi_n$, $n\geq0$. If  $\left<\varphi_0,\Psi_0\right>=1$, then
\be \left<\varphi_n,\Psi_m\right>=\delta_{n,m}, \label{A4}\en
for all $n, m\geq0$. Hence $\F_\Psi=\{\Psi_{ n}, \,n\geq0\}$ and
$\F_\varphi=\{\varphi_{ n}, \,n\geq0\}$ are biorthonormal.

The analogy with ordinary bosons suggests us to consider the following:

\vspace{2mm}

{\bf Assumption $\D$-pb 3.--}  $\F_\varphi$ is a basis for $\Hil$.

\vspace{1mm}

This is equivalent to requiring that $\F_\Psi$ is a basis for $\Hil$ as well. However, several  physical models show that $\F_\varphi$ is {\bf not} always a basis for $\Hil$, but it is still total in $\Hil$: if $f\in\Hil$ is orthogonal to $\varphi_n$, for all $n$, then $f=0$. For this reason we have adopted the following weaker version of  Assumption $\D$-pb 3, \cite{baginbagbook}:

\vspace{2mm}

{\bf Assumption $\D$-pbw 3.--}  For some subspace $\G$ dense in $\Hil$, $\F_\varphi$ and $\F_\Psi$ are $\G$-quasi bases.

\vspace{2mm}
This means that, for all $f$ and $g$ in $\G$,
\be
\left<f,g\right>=\sum_{n\geq0}\left<f,\varphi_n\right>\left<\Psi_n,g\right>=\sum_{n\geq0}\left<f,\Psi_n\right>\left<\varphi_n,g\right>,
\label{A4b}
\en
which can be seen as a weak form of the resolution of the identity, restricted to $\G$.

The families $\F_\varphi$ and $\F_\Psi$ can be used to define two densely defined operators $S_\varphi$ and $S_\Psi$ via their
action respectively on  $\F_\Psi$ and $\F_\varphi$: \be
S_\varphi\Psi_{ n}=\varphi_{ n},\qquad
S_\Psi\varphi_{ n}=\Psi_{\bf n}, \label{213}\en for all $ n$. These operators play a very import role in the analysis of pseudo-bosons, since they map $\F_\varphi$ into $\F_\Psi$ and vice-versa, and define new scalar products in $\Hil$ is terms of which, for instance, the (new) adjoint of $b$ turns out to coincide with $a$. These, and many other aspects which are not relevant here, can be found in \cite{bag2020JPA,baginbagbook}. 

\subsection{Bi-coherent states}\label{sectbcs}

Let us consider two biorthogonal families of vectors, $\F_{\tilde\varphi}=\{\tilde\varphi_n\in\Hil, \, n\geq0\}$ and $\F_{\tilde\Psi}=\{\tilde\Psi_n\in\Hil, \, n\geq0\}$ which are $\G$
-quasi bases for some dense subset of $\Hil$, $\G$, see (\ref{A4b}). Consider an increasing sequence of real numbers $\alpha_n$ satisfying the  inequalities $0=\alpha_0<\alpha_1<\alpha_2<\ldots$. We call $\overline\alpha$ the limit of $\alpha_n$ for $n$ diverging, which coincides with $\sup_n\alpha_n$. We further consider two operators, $A$ and $B^\dagger$, which act as lowering operators respectively on $\F_{\tilde\varphi}$ and $\F_{\tilde\Psi}$ in the following way:
\be
A\,\tilde\varphi_n=\alpha_n\tilde\varphi_{n-1}, \qquad B^\dagger\,\tilde\Psi_n=\alpha_n\tilde\Psi_{n-1},
\label{20}\en
for all $n\geq1$, with $A\,\tilde\varphi_0=B^\dagger\,\tilde\Psi_0=0$. These are the lowering equations which replace those in (\ref{A3}), which can be recovered if $\alpha_n=\sqrt{n}$ and if $A$ and $B$ obey (\ref{A1}). Then the following theorem holds, \cite{bagproc}:

\begin{thm}\label{theo1}
	Assume that four strictly positive constants $A_\varphi$, $A_\Psi$, $r_\varphi$ and $r_\Psi$ exist, together with two strictly positive sequences $M_n(\varphi)$ and $M_n(\Psi)$, for which
	\be
	\lim_{n\rightarrow\infty}\frac{M_n(\varphi)}{M_{n+1}(\varphi)}=M(\varphi), \qquad \lim_{n\rightarrow\infty}\frac{M_n(\Psi)}{M_{n+1}(\Psi)}=M(\Psi),
	\label{21}\en
	where $M(\varphi)$ and $M(\Psi)$ could be infinity, and such that, for all $n\geq0$,
	\be
	\|\tilde\varphi_n\|\leq A_\varphi\,r_\varphi^n M_n(\varphi), \qquad \|\tilde\Psi_n\|\leq A_\Psi\,r_\Psi^n M_n(\Psi).
	\label{22}\en
	Then, putting $\alpha_0!=1$ and $\alpha_k!=\alpha_1\alpha_2\cdots\alpha_k$, $k\geq1$, the following series:
	\be
	N(|z|)=\left(\sum_{k=0}^\infty\frac{|z|^{2k}}{(\alpha_k!)^2}\right)^{-1/2},
	\label{23}\en
	\be
	\varphi(z)=N(|z|)\sum_{k=0}^\infty\frac{z^k}{\alpha_k!}\tilde\varphi_k,\qquad \Psi(z)=N(|z|)\sum_{k=0}^\infty\frac{z^k}{\alpha_k!}\tilde\Psi_k,
	\label{24}\en
	are all convergent inside the circle $C_\rho(0)$ in $\mathbb{C}$ centered in the origin of the complex plane and of radius $\rho=\overline\alpha\,\min\left(1,\frac{M(\varphi)}{r_\varphi},\frac{M(\Psi)}{r_\Psi}\right)$. Moreover, for all $z\in C_\rho(0)$,
	\be
	A\varphi(z)=z\varphi(z), \qquad B^\dagger \Psi(z)=z\Psi(z).
	\label{25}\en
	Suppose further that a measure $d\lambda(r)$ does exist such that
	\be
	\int_0^\rho d\lambda(r) \,r^{2k}=\frac{(\alpha_k!)^2}{2\pi},
	\label{26}\en
	for all $k\geq0$. Then, putting $z=re^{i\theta}$ and calling $d\nu(z,\overline z)=N(r)^{-2}d\lambda(r)d\theta$, we have
	\be
	\int_{C_\rho(0)}\left<f,\Psi(z)\right>\left<\varphi(z),g\right>d\nu(z,\overline z)=
	\int_{C_\rho(0)}\left<f,\varphi(z)\right>\left<\Psi(z),g\right>d\nu(z,\overline z)=
	\left<f,g\right>,
	\label{27}\en
	for all $f,g\in\G$.
	
\end{thm}

Some comments are in order: first we observe that, contrarily to what happens for ordinary coherent states, \cite{aagbook,didier,gazeaubook},  the norms of the vectors $\tilde\varphi_n$ and $\tilde\Psi_n$ need not being uniformly bounded, here. On the contrary, they can diverge rather fast with $n$, see (\ref{22}). Of course, this is reflected by the fact that bi-coherent states of this kind only exist inside $C_\rho(0)$. We also observe that no mention is made here to the {\em displacement-like} operators usually relevant in connection with ordinary coherent states\footnote{We are referring here to the unitary operator $U(z)=e^{\overline{z}\, c-z c^\dagger}$, where $[c,c^\dagger]=\1$, which should be replaced here, for instance, by $e^{\overline{z}\, a-z b}$.}. This is, indeed, a non trivial aspect of the theory of bi-coherent states, discussed at length in \cite{bagspringer}. Another relevant comment here is that Theorem \ref{theo1} is given in an Hilbert space. Indeed, $\tilde\varphi_n$ and $\tilde\Psi_n$ have  finite norms, as well as the vectors $\varphi(z)$ and $\Psi(z)$. However, in some particular models, $\|\tilde\varphi_n\|=\|\tilde\Psi_n\|=\infty$, for all (or some) $n$. Hence, working in $\Hil$ is not the most appropriate choice, of course. We have discussed this situation in \cite{bagJPCS2021} and in \cite{bagspringer}, and it is at the basis of what we will discuss in Section \ref{sectBCS}. Last but not least, if $A$ and $B$ are pseudo-bosonic, then $\alpha_n=\sqrt{n}$ and, therefore, $\overline\alpha=\infty$ and $C_\rho(0)$ coincides with the whole complex plane, at least if $M(\varphi)$ and $M(\Psi)$ are both non zero.

\section{A general class of pseudo-bosonic operators}\label{sectgenclasspbs}

The main aim of this paper is to consider first order differential operators of the form
\be
a=\alpha_a(x)\,\frac{d}{dx}+\beta_a(x), \qquad b=-\frac{d}{dx}\,\alpha_b(x)+\beta_b(x), 
\label{41}\en
for some suitable functions $\alpha_j(x)$ and $\beta_j(x)$, $j=a,b$. In what follows we will only need to compute the first and the second derivatives of these functions. However, in all the examples considered in the rest of the paper these are $C^\infty$ functions, as in \cite{bagJPCS2021}. This is what happens also in all the existing literature, in our knowledge. In particular, for ordinary bosons ($a=c=\frac{1}{\sqrt{2}}(\frac{d}{dx}+x)$ and $b=c^\dagger=\frac{1}{\sqrt{2}}(-\frac{d}{dx}+x)$), we have $\alpha_a(x)=\alpha_b(x)=\frac{1}{\sqrt{2}}$, while   $\beta_a(x)=\beta_b(x)=\frac{1}{\sqrt{2}}\,x$. For the shifted harmonic oscillator, see \cite{baginbagbook} and references therein, we have $a=c+\alpha\1$ and $b=c^\dagger+\beta\1$, for some complex $\alpha$ and $\beta$ with $\alpha\neq \overline{\beta}$, and therefore $\alpha_a(x)=\alpha_b(x)=\frac{1}{\sqrt{2}}$ as before, while   $\beta_a(x)=\frac{1}{\sqrt{2}}\,x+\alpha$ and $\beta_b(x)=\frac{1}{\sqrt{2}}\,x+\beta$. Another interesting quantum mechanical system which have been considered in this context is the Swanson model, see again \cite{baginbagbook} and references therein, where
$$
a=\frac{1}{\sqrt{2}}\left(e^{-i\theta}\,\frac{d}{dx}+e^{i\theta}x\right), \qquad b=\frac{1}{\sqrt{2}}\left(-e^{-i\theta}\,\frac{d}{dx}+e^{i\theta}x\right).
$$
In this case, $\alpha_a(x)=\alpha_b(x)=\frac{e^{-i\theta}}{\sqrt{2}}$, while $\beta_a(x)=\beta_b(x)=\frac{e^{i\theta}x}{\sqrt{2}}$.

More recently, \cite{bagJPCS2021,bagspringer}, a rather general class of pseudo-bosonic operators have been considered, where $A=\frac{d}{dx}+w_A(x)$ and $B=-\frac{d}{dx}+w_B(x)$. In this case $\alpha_a(x)=\alpha_b(x)=1$, while $w_A(x)$ and $w_B(x)$ have been called {\em pseudo-bosonic superpotentials} (PBSs) and they must satisfy $(w_A(x)+w_B(x))'=1$, where the prime stands for the first $x$-derivative. In particular, in this last example, different choices of $C^\infty$ functions  $w_A(x)$ and $w_B(x)$ give rise to different families of functions, $\varphi_{n}(x)$ and $\Psi_{n}(x)$, constructed in analogy with (\ref{A2}), which may, or may not, be square-integrable. However, see \cite{bagJPCS2021}, we have proven the following result:

\begin{prop}\label{prop1}
	If $w_A(x)$ and $w_B(x)$ are $C^\infty$ PBSs, then $\varphi_{n}(x)\, \overline{\Psi_m(x)}\in\Lc^1(\mathbb{R})$ and  $\langle\Psi_m,\varphi_n\rangle=\delta_{n,m}$, for all $n,m\geq0$.
\end{prop}

In this case the functions $\varphi_{n}(x)$ and $\Psi_{n}(x)$ are called {\em compatible}, in the sense of the so-called PIP-spaces, \cite{pip}. In this perspective it is useful to recall that two functions $h_1(x)\in\Lc^p(\mathbb{R})$ and $h_2(x)\in\Lc^q(\mathbb{R})$ can be multiplied producing a third function $h(x)=h_1(x)h_2(x)$ which is integrable, $h(x)\in\Lc^1(\mathbb{R})$, if $\frac{1}{p}+\frac{1}{q}=1$. Hence, a sort of scalar product can be defined also for these pairs of function. But, rather than using the term {\em scalar product}, we prefer to adopt a different terminology, and call this a {\em compatibility form}. It is clear that, for those functions which are compatible, a generalized notion of biorthonormality can be introduced.

\vspace{1mm}
   
In what follows, we are interested in extending the results in \cite{bagJPCS2021} to the operators in (\ref{41}). In particular, we want to discuss the following aspects:

\begin{enumerate}
	\item do these operators obey pseudo-bosonic commutation rules? 
	\item do they produce biorthonormal families of vectors?
	\item do these vectors belong to $\ltwo$? Or, in case they do not, are these families compatible?
	\item are these vectors (generalized) eigenvectors of some particular operator? 
	\item are $a$ and $b$ connected to some families of bi-coherent states? 
	\item do these bi-coherent states produce some sort of resolution of the identity?
\end{enumerate}

To  answer these questions we first compute $[a,b]$ on some sufficiently regular function $f(x)$. It is important to stress that, as already mentioned, in the analysis proposed in this paper the role of $\ltwo$ is not essential, since we are more interested in compatible pairs of functions, rather than in square integrable ones. For this reason in what follows we will not impose $f(x)$ to belong to some suitable subspace of $\ltwo$, but only to be regular enough to admit all the computations we need to perform on it. More explicitly , we will assume $f(x)$ to be at least $C^2$. Of course, this requirement could be relaxed if we interpret $\frac{d}{dx}$ as the weak derivative, but this will not be done here.
An easy computation shows that, under this mild condition on $f(x)$, $[a,b]f(x)$ does make sense, and $[a,b]f(x)=f(x)$ if $\alpha_j(x)$ and $\beta_j(x)$, $j=a,b$, satisfy the following equalities
\be
\left\{
\begin{array}{ll}
	\alpha_a(x)\alpha_b'(x)=\alpha_a'(x)\alpha_b(x), \\
	\alpha_a(x)\beta_b'(x)+	\alpha_b(x)\beta_a'(x)=1+\alpha_a(x)\alpha_b''(x).\\
\end{array}
\right.
\label{42}\en
It is easy to check that all the examples listed at the beginning of this section satisfy indeed these two conditions, in agreement with their nature of pseudo-bosonic operators. In particular the first equation in (\ref{42}) is always true for all constant choice of $\alpha_a(x)$ and $\alpha_b(x)$. Moreover, in this case, the second equation in (\ref{42}) can be rewritten as $(\alpha_a\beta_b(x)+\alpha_b\beta_a(x))'=1$, which means that $\alpha_a\beta_b(x)+\alpha_b\beta_a(x)=x+k$, for some constant $k$. This is essentially the situation described in terms of the PBSs $w_A(x)$ and $w_B(x)$ in \cite{bagJPCS2021,bagspringer}. Incidentally it is also clear that, if $\alpha_a(x)=\alpha_a\neq0$, constant, then (\ref{42}) implies that  $\alpha_a(x)\alpha_b'(x)=\alpha_a\alpha_b'(x)=0$, which means that $\alpha_b(x)$ must also be constant. For this reason, to avoid going back to PBSs, in the rest of this paper we will  mainly focus our interest on the situation in which both $\alpha_a(x)$ and $\alpha_b(x)$ depend on $x$ in a non trivial way. Moreover, it is convenient for what follows to assume that they are never zero: $\alpha_j(x)\neq0$, $\forall x\in\mathbb{R}$, $j=a,b$. 

Under this assumption it is easy to deduce the vacua of $a$ and of $b^\dagger$, as in Section \ref{sectpbs}. In what follows the adjoint of $a$ and $b$ are operators which can be formally deduced by the standard formula $\langle X^\dagger f,g\rangle=\langle f,Xg\rangle$, for suitable $f$ and $g$. However here, recalling that we are not really interested in the role of $\ltwo$, this formula may appear {\em strange}, since the meaning of the scalar product must still be understood. For this reason, we simply call $a^\dagger$ and $b^\dagger$ the following operators:
\be
a^\dagger=-\frac{d}{dx}\,\overline{\alpha_a(x)}+\overline{\beta_a(x)}, \qquad b^\dagger=\overline{\alpha_b(x)}\,\frac{d}{dx}+\overline{\beta_b(x)},
\label{43}\en
since these are indeed the formal adjoints\footnote{These {\em formal} adjoints could be made {\em rigorous} with a proper choice of the domains of the operators involved. But, as already mentioned, the role of $\ltwo$ is not so crucial in our settings. For this reason, we consider the operators in (\ref{43}) as part of our building blocks.} of $a$ and $b$.

The vacua of $a$ and $b^\dagger$ are the solutions of $a\varphi_0(x)=0$ and $b^\dagger\psi_0(x)=0$, which are easily found:
\be
\varphi_0(x)=N_\varphi \exp\left\{-\int\frac{\beta_a(x)}{\alpha_a(x)}\,dx\right\}, \qquad \psi_0(x)=N_\psi \exp\left\{-\int\frac{\overline{\beta_b(x)}}{\overline{\alpha_b(x)}}\,dx\right\},
\label{44}\en
and are well defined under our assumptions on $\alpha_j(x)$ and $\beta_j(x)$. Here $N_\varphi$ and $N_\psi$ are normalization constants which will be fixed later. If we now introduce $\varphi_n(x)$ and $\psi_n(x)$ as in (\ref{A2}),
\be
\varphi_n(x)=\frac{1}{\sqrt{n!}}\,b^n\varphi_0(x),\qquad \psi_n(x)=\frac{1}{\sqrt{n!}}\,{a^\dagger}^n\psi_0(x), \label{45}\en
$n\geq0$, we can prove the following:

\begin{prop}\label{propstati}
	Calling $\theta(x)=\alpha_a(x)\beta_b(x)+\alpha_b(x)\beta_a(x)$ we have
	\be
	\varphi_n(x)=\frac{1}{\sqrt{n!}}\,\pi_n(x)\varphi_0(x), \qquad \psi_n(x)=\frac{1}{\sqrt{n!}}\,\sigma_n(x)\varphi_0(x),
	\label{46}\en
	$n\geq0$, where $\pi_n(x)$ and $\sigma_n(x)$ are defined recursively as follows:
	\be
	\pi_0(x)=\sigma_0(x)=1,
	\label{47}\en
	and
	\be
	\pi_n(x)=\left(\frac{\theta(x)}{\alpha_a(x)}-\alpha_b'(x)\right)\pi_{n-1}(x)-\alpha_b(x)\pi_{n-1}'(x),
	\label{48}\en
	\be
	\sigma_n(x)=\overline{\left(\frac{\theta(x)}{\alpha_b(x)}-\alpha_a'(x)\right)}\,\sigma_{n-1}(x)-\overline{\alpha_a(x)}\,\sigma_{n-1}'(x),
	\label{49}\en
	$n\geq1$.
\end{prop}

\begin{proof}
	We prove the statement for $\varphi_n(x)$ by induction on $n$. The proof for $\psi_n(x)$ is similar.
	
	The statement is trivially true for $n=0$. Now, let us assume that the statement is true for $n-1$: hence $\varphi_{n-1}(x)=\frac{1}{\sqrt{(n-1)!}}\,\pi_{n-1}(x)\varphi_0(x)$, with $\pi_{n-1}(x)$ related to $\pi_{n-2}(x)$ as in (\ref{48}), and let us prove that a similar relation holds for $\pi_n(x)$. Indeed we have, after few manipulations,
	$$
	\sqrt{n!}\,\varphi_n(x)=\sqrt{(n-1)!}\,b\,\varphi_{n-1}(x)=b\,\pi_{n-1}(x)\varphi_0(x)=\left(-\frac{d}{dx}\,\alpha_b(x)+\beta_b(x)\right)\pi_{n-1}(x)\varphi_0(x)=
	$$
	$$
	=\left(-\alpha_b'(x)\pi_{n-1}(x)-\alpha_b(x)\pi_{n-1}'(x)+\frac{\beta_a(x)\alpha_b(x)\pi_{n-1}(x)}{\alpha_a(x)}+\beta_b(x)\pi_{n-1}(x)\right)\varphi_0(x)=
	$$
	$$
	=\left[\left(\frac{\theta(x)}{\alpha_a(x)}-\alpha_b'(x)\right)\pi_{n-1}(x)-\alpha_b(x)\pi_{n-1}'(x)\right]\varphi_0(x)=\pi_n(x)\varphi_0(x),
	$$
	which is what we had to prove.

\end{proof}

\subsection{A special case: constant $\alpha_j(x)$}\label{subsectionconstantalphas}

Let us see what happens if, in particular, $\alpha_a(x)=\alpha_a$ and $\alpha_b(x)=\alpha_b$. Of course, in this case, $\alpha_a(x)$ and $\alpha_b(x)$ are always different from zero, at least if $\alpha_a\alpha_b\neq0$. Formulas (\ref{48}) and (\ref{49}) simplify significantly now since, in particular, as we have already deduced before, $\theta(x)=\alpha_a\beta_b(x)+\alpha_b\beta_a(x)=x+k$. Hence we find
\be
\pi_n(x)=\frac{1}{\alpha_a}(x+k)\,\pi_{n-1}(x)-\alpha_b\,\pi_{n-1}'(x),\qquad \sigma_n(x)=\frac{1}{\overline\alpha_a}(x+\overline k)\,\sigma_{n-1}(x)-\overline\alpha_a\,\sigma_{n-1}'(x),
\label{410}\en
The case $\alpha_a=\alpha_b=1$ has been considered in \cite{bagJPCS2021}, while $\alpha_a=\alpha_b=\frac{1}{\sqrt{2}}$ is discussed in \cite{bagspringer}. If $\alpha_a$ is not necessarily equal to $\alpha_b$, similar conclusions to those deduced in \cite{bagJPCS2021,bagspringer} can still be deduced. In particular from (\ref{410}) we find that
\be
\pi_n(x)=\sqrt{\left(\frac{\alpha_b}{2\alpha_a}\right)^n}H_n\left(\frac{x+k}{\sqrt{2\alpha_a\alpha_b}}\right), \qquad \sigma_n(x)=\sqrt{\left(\frac{\overline\alpha_b}{2\overline\alpha_a}\right)^n}H_n\left(\frac{x+\overline k}{\sqrt{2\overline\alpha_a\overline\alpha_b}}\right).
\label{411}\en
Here $H_n(x)$ is the $n$-th Hermite polynomial, and the square root of the complex quantities are taken to be their principal determinations.

To prove formula (\ref{411}) for $\pi_n(x)$ we use induction on $n$. The statement is clearly true for $n=0$. Let us now suppose that it is also true for $n-1$. This means that
$$
\pi_{n-1}(x)=\sqrt{\left(\frac{\alpha_b}{2\alpha_a}\right)^{n-1}}H_{n-1}\left(\frac{x+k}{\sqrt{2\alpha_a\alpha_b}}\right).
$$
To check that the same formula holds for $n$, we use (\ref{410}):
$$
\pi_n(x)=\sqrt{\left(\frac{\alpha_b}{2\alpha_a}\right)^{n-1}}\left(\frac{1}{\alpha_a}(x+k)\,H_{n-1}\left(\frac{x+k}{\sqrt{2\alpha_a\alpha_b}}\right)-\alpha_b\,\frac{d}{dx}H_{n-1}\left(\frac{x+k}{\sqrt{2\alpha_a\alpha_b}}\right)\right)=
$$
$$
=\sqrt{\left(\frac{\alpha_b}{2\alpha_a}\right)^{n}}\left[2yH_{n-1}(y)-H_{n-1}'(y)\right]_{y=\frac{x+k}{\sqrt{2\alpha_a\alpha_b}}}=\sqrt{\left(\frac{\alpha_b}{2\alpha_a}\right)^n}H_n\left(\frac{x+k}{\sqrt{2\alpha_a\alpha_b}}\right),
$$
after some minor manipulations, and using the well known identity  for Hermite polynomials $H_n(y)=2yH_{n-1}(y)-H_{n-1}'(y)$.

As for the functions in (\ref{44}) we get $\varphi_0(x)=N_\varphi \exp\left\{-\,\frac{1}{\alpha_a}\int\beta_a(x)\,dx\right\}$, and $\psi_0(x)=N_\psi \exp\left\{-\frac{1}{\overline{\alpha_b}}\int\overline{\beta_b(x)}\,dx\right\},$ where $\beta_a(x)$ and $\beta_b(x)$ are only required to satisfy the condition $\alpha_a\beta_b(x)+\alpha_b\beta_a(x)=x+k$. Now, extending what proved in \cite{bagJPCS2021}, it is possible to deduce that $\varphi_{n}(x)\, \overline{\Psi_m(x)}\in\Lc^1(\mathbb{R})$, for all $n,m\geq0$, as in Proposition \ref{prop1} above. The proof is based on the fact that $\varphi_{n}(x)\, \overline{\Psi_m(x)}$ is (a part some normalization constants), the product of a polynomial of degree $n+m$ times the following exponential
$$
\exp\left\{-\,\int\left(\frac{\beta_a(x)}{\alpha_a}+\frac{\beta_b(x)}{\alpha_b}\right)\,dx\right\}=\exp\left\{-\frac{1}{\alpha_a\alpha_b}\,\int\theta(x)\,dx\right\}=$$
$$=\exp\left\{-\frac{1}{\alpha_a\alpha_b}\,\int(x+k)\,dx\right\}=\exp\left\{-\frac{1}{\alpha_a\alpha_b}\left(\frac{x^2}{2}+kx+\tilde k\right)\right\},
$$
for some integration constant $\tilde k$. Notice that this is a gaussian term whenever $\alpha_a\alpha_b>0$. In \cite{bagJPCS2021} the biorthonormality of $\F_\varphi=\{\varphi_n(x)\}$ and  $\F_\psi=\{\psi_n(x)\}$ is discussed. Of course, calling these sets {\em biorthonormal} is a little abuse of language, since there is no guarantee that $\varphi_n(x)$ and $\psi_m(x)$ are square-integrable, even if the product of the two can be integrated: we see that, as already pointed out, the compatibility form is well defined, and it extends the scalar product to non necessarily square-integrable functions.

\section{A general example}\label{sectAGE}

The situation we will now consider in when $\alpha_a(x)=\alpha_b(x)=\alpha(x)$, where $\alpha(x)\neq0$ for all $x\in\mathbb{R}$. In this case the first equation in (\ref{42}) is automatically true, independently of the particular form of $\alpha(x)$. The second equation becomes $(\beta_a(x)+\beta_b(x))'=\frac{1}{\alpha(x)}+\alpha''(x)$, which produces 
\be
\beta_a(x)+\beta_b(x)=\int\frac{dx}{\alpha(x)}+\alpha'(x).
\label{51}\en
From now on we will identify $\beta_a(x)$ and $\beta_b(x)$ as follows:
\be
\beta_a(x)=\int\frac{dx}{\alpha(x)}, \qquad \beta_b(x)=\alpha'(x).
\label{52}\en
Of course, other possible choices exist. The easiest alternative is when the role of $\beta_a(x)$ and $\beta_b(x)$ are exchanged. But we could also consider $\beta_a(x)=\int\frac{dx}{\alpha(x)}+\Phi(x)$ and $ \beta_b(x)=\alpha'(x)-\Phi(x)$, for all possible choices of (sufficiently regular) $\Phi(x)$. However, we will take $\Phi(x)=0$ in what follows. Similarly, we will also fix to zero all the integration constants, except when explicitly stated. The function $\theta(x)$ introduced in Proposition \ref{propstati} becomes $\theta(x)=\alpha(x)(\beta_a(x)+\beta_b(x))$, so that
\be
\theta(x)=\alpha(x)\left(\int\frac{dx}{\alpha(x)}+\alpha'(x)\right),
\label{53}\en
which, when replaced in (\ref{48}), produces the following sequence of functions: $\pi_0(x)=1$ and
\be
\pi_n(x)=\left(\int\frac{dx}{\alpha(x)}\right)\pi_{n-1}(x)-\alpha(x)\pi_{n-1}'(x).
\label{54}\en
Calling $\rho(x)=\int\frac{dx}{\alpha(x)}$ we can rewrite (\ref{54}) in the following alternative way:
\be
\pi_n(x)=\rho(x)\pi_{n-1}(x)-\frac{1}{\rho'(x)}\pi_{n-1}'(x),
\label{55}\en
$n\geq1$, which can be used to deduce the following expression for $\pi_n(x)$:
\be
\pi_n(x)=\frac{1}{\sqrt{2^n}}\,H_n\left(\frac{\rho(x)}{\sqrt{2}}\right),
\label{56}\en
for all $n\geq0$. The proof is similar to that given in Section \ref{subsectionconstantalphas}, and will not be repeated here. 

\vspace{2mm}

{\bf Remark:--} It is worth stressing that (\ref{56}) returns the first equation in (\ref{411}) if $\alpha(x)=\alpha$, constant in $x$. Indeed, in this case, from (\ref{411}) we deduce that $\pi_n(x)=\sqrt{\frac{1}{2^n}}\,H_n\left(\frac{x+k}{\sqrt{2}\alpha}\right)$, while $\rho(x)=\frac{1}{\alpha}\int\,dx=\frac{x+k}{\alpha}$, for some integration constant $k$. Hence (\ref{56}) produces the same result.

\vspace{2mm}

Quite often we will take $\alpha(x)$ real. Then, using (\ref{52}), $\beta_b(x)$ is also real, while $\beta_a(x)$ is real if the integration constant is chosen to be real, as we will do always here. Under these conditions the functions $\sigma_n(x)$ coincide with $\pi_n(x)$: $\sigma_n(x)=\pi_n(x)$, $\forall n\geq0$.

\vspace{2mm}

{\bf Remark:--} The appearance of the Hermite polynomials in the formulas for the various $\varphi_n(x)$ and $\psi_n(x)$, here and in other papers, see \cite{bagJPCS2021} and \cite{bagspringer} in particular, is due, we believe, to the particular form of the pseudo-bosonic commutation relations in (\ref{A1}). The fact that these are deformations of the CCR, which are connected to Hermite polynomials, is reflected by the appearance of Hermite polynomials of {\em more elaborated arguments} in our {\em more elaborated context}.

\vspace{2mm}

As for the vacua in (\ref{44}), using the fact that $\alpha_a(x)=\alpha_b(x)=\alpha(x)$, together with formulas (\ref{52}), we deduce that
\be
\varphi_0(x)=N_\varphi\exp\left\{-\frac{1}{2}(\rho(x))^2\right\}, \qquad \psi_0(x)=\frac{N_\psi}{\overline \alpha(x)},
\label{57}\en
or simply $\psi_0(x)=\frac{N_\psi}{\alpha(x)}$ if $\alpha(x)$ is real. Putting all together we conclude that
\be
\varphi_n(x)=\frac{N_\varphi}{\sqrt{2^nn!}}H_n\left(\frac{\rho(x)}{\sqrt{2}}\right)e^{-\left(\frac{\rho(x)}{\sqrt{2}}\right)^2}, \qquad \psi_n(x)=\frac{N_\psi}{\sqrt{2^nn!}}H_n\left(\frac{\rho(x)}{\sqrt{2}}\right)\frac{1}{ \alpha(x)},
\label{58}\en
where we have also assumed (for $\psi_n(x)$) that $\alpha(x)$ is real, to simplify the notation.

It is now very easy to prove that, under very mild assumption on $\alpha(x)$, the families $\F_\varphi$ and $\F_\psi$ are compatible and biorthonormal (in our slightly extended meaning), even when the functions $\varphi_n(x)$ or $\psi_n(x)$ do not belong to $\ltwo$. To prove these claims, it is useful to assume that $\rho(x)$ is increasing in $x$ and that, calling $s=\frac{\rho(x)}{\sqrt{2}}$, $s\rightarrow\pm\infty$ when $x\rightarrow\pm\infty$. It is clear then that $\rho$ can be inverted, and that $x=\rho^{-1}(\sqrt{2}s)$. Since $\rho'(x)=\frac{1}{\alpha(x)}$, it follows that $\rho(x)$ is always increasing if $\alpha(x)>0$. However, this is not enough to ensure that $s$ diverges with $x$, and therefore must also be required.

 Now, to prove that $\varphi_n(x)$ and $\psi_m(x)$ are compatible (and biorthonormal), we compute the { compatibility form}:
$$
\langle\psi_m,\varphi_n\rangle=\frac{\overline N_\psi N_\varphi}{\sqrt{2^{n+m}\,n!\,m!}}\int_{-\infty}^\infty H_m\left(\frac{\rho(x)}{\sqrt{2}}\right)H_n\left(\frac{\rho(x)}{\sqrt{2}}\right)e^{-\left(\frac{\rho(x)}{\sqrt{2}}\right)^2}\frac{dx}{ \alpha(x)}.
$$
This integral can be easily rewritten in terms of $s$. In fact, recalling the definition of $\rho(x)$, we first observe that $ \frac{ds}{dx}=\frac{1}{\sqrt{2}\,\alpha(x)}$, so that $\frac{dx}{\alpha(x)}=\sqrt{2}\,ds$. Hence we have
$$
\langle\psi_m,\varphi_n\rangle=\frac{\overline N_\psi N_\varphi}{\sqrt{2^{n+m-1}\,n!\,m!}}\int_{-\infty}^\infty H_m(s)H_n(s)e^{-s^2}ds=\sqrt{2\pi}\,\overline N_\psi N_\varphi\,\delta_{n,m},
$$
which returns
\be
\langle\psi_m,\varphi_n\rangle=\delta_{n,m}, \qquad \mbox{if }\qquad  \overline N_\psi N_\varphi=\frac{1}{\sqrt{2\pi}},
\label{59}\en
as will be assumed in the rest of this section. This is what we had to prove

\subsection{Quasi-basis nature of $\F_\varphi$ and $\F_\psi$}\label{sectquasibases}

It is clear that, in general, $\F_\varphi$ and $\F_\psi$ are not bases for $\ltwo$. This is an obvious consequence of the fact that it is not granted at all that their elements are square integrable. However, in recent papers, this aspect has been circumvented by replacing the notion of basis with that of quasi-basis, as in (\ref{A4b}), for instance. The idea is that, since $\varphi_n(x)$ or $\psi_n(x)$, or both, can have {\em bad} behaviour\footnote{In particular, they can be converging to zero too slowly, or not be converging to zero at all!}, it is convenient to look for a resolution of the identity only on some set of {\em particularly good} functions. This is, in fact, not very different from what is done in distribution theory, \cite{gel}, where distributions are mathematically {\em complicated} objects which acquire a rigorous meaning when considered in pairs with some sets of good functions.

With this in mind, let us introduce the set
\be
\E=\left\{h(s)\in\ltwo:\, h_-(s):=h(\rho^{-1}(\sqrt{2}s))\,e^{s^2/2}\in\ltwo\right\} 
\label{510}\en
This set is dense in $\ltwo$. Indeed, it contains the set $\D(\mathbb{R})$ of all the compactly supported $C^\infty$ functions. In fact, it is easy to see that $h_-(s)$ is compactly supported and continuous. Hence the integral of its square modulus exists. In particular, if $\rho^{-1}$ is $C^\infty$, then $h_-(s)\in\D(\mathbb{R})$ for all $h(x)\in\D(\mathbb{R})$. Another useful result is that, if $h(x)\in\E$, then the function $h_+(s):={h(\rho^{-1}(\sqrt{2}s))}\,\alpha(\rho^{-1}(\sqrt{2}s))\,e^{-s^2/2}\in\ltwo$ as well, at least under very general conditions on $\alpha(x)$. This is because
$
|h_+(s)|^2=|h_-(s)|^2 |g(s)|^2,
$
where $g(s)=\alpha(\rho^{-1}(\sqrt{2}s))\,e^{-s^2}$. Now, it is sufficient that $g(s)\in\Lc^\infty(\mathbb{R})$ to conclude that $h_+(s)\in\ltwo$. But, because of the presence of $e^{-s^2}$ in $g(s)$, this is true for many choices of $\alpha(x)$, as we will see later, when concrete choices will be considered. However, even if $\alpha(x)$ diverges very fast, if $h(x)\in\D(\mathbb{R})$ then $h_+(s)\in\ltwo$ anyhow, which is what we will use in the following.

\begin{thm}\label{thm1}\label{thquasibases}
	$(\F_\varphi,\F_\psi)$ are $\E$-quasi bases.
\end{thm}

\begin{proof}
	Let us take $f(x), g(x)\in\E$. It is possible to check that the following equalities hold:
	\be
	\langle f,\varphi_n\rangle=N_\varphi\,\pi^{1/4}\sqrt{2}\langle f_+,e_n\rangle, \qquad \langle \psi_n,g\rangle=\overline N_\psi\,\pi^{1/4}\sqrt{2}\langle e_n,g_-\rangle.
	\label{511}\en
	Here $e_n(s)=\frac{1}{\sqrt{2^nn!\sqrt{\pi}}}\,H_n(s)e^{-s^2/2}$ is the $n$-th eigenstate of the quantum harmonic oscillator, while $f_+(s)$ and $g_-(s)$ should be constructed from $f(s)$ and $g(s)$ as shown before. 
	The equalities in (\ref{511}) show, in particular, that the pairs $(f(x),\varphi_n(x))$ and $(g(x),\psi_n(x))$ are compatible, $\forall n\geq0$, since all the functions involved in the right-hand sides of the equalities in (\ref{511}), $e_n(s)$, $f_+(s)$ and $g_-(s)$, are square integrable\footnote{Stated differently, one could say that, e.g., $\langle f,\varphi_n\rangle$ is the compatibility form between $f$ and $\varphi_n$, while $\langle f_+,e_n\rangle$ is the scalar product between $f_+$ and $e_n$. And they are connected.}. It is well known that the set $\F_e=\{e_n(s), \,n\geq0\}$ is an orthonormal basis for $\ltwo$.
	
	The proof of these identities is based on the change of variable $s=\frac{\rho(x)}{\sqrt{2}}$, which has already been used before, to prove (\ref{59}). We can now use (\ref{511}) as follows:
	$$
	\sum_{n=0}^\infty \langle f,\varphi_n\rangle\langle \psi_n,g\rangle=\overline N_\psi\,N_\varphi\,2\sqrt{\pi}\sum_{n=0}^\infty \langle f_+,e_n\rangle\langle e_n,g_-\rangle=\sqrt{2}\langle f_+,g_-\rangle,
	$$
	using (\ref{59}) and the Parceval identity (i.e., the resolution of the identity) for $\F_e$. Next we have
	$$
	\langle f_+,g_-\rangle=\int_{-\infty}^\infty \overline{f_+(s)}\,g_-(s)\,ds=\int_{-\infty}^\infty \overline{f(\rho^{-1}(\sqrt{2}s))}\,\alpha(\rho^{-1}(\sqrt{2}s))e^{-s^2/2}g(\rho^{-1}(\sqrt{2}s))e^{s^2/2}\,ds=
	$$
	$$
	=\int_{-\infty}^\infty \overline{f(\rho^{-1}(\sqrt{2}s))}\,\alpha(\rho^{-1}(\sqrt{2}s))g(\rho^{-1}(\sqrt{2}s))\,ds=\frac{1}{\sqrt{2}}\langle f,g\rangle,
	$$
	introducing the new variable $x=\rho^{-1}(\sqrt{2}s)$ in the integral. Summarizing we have
	$$
	\sum_{n=0}^\infty \langle f,\varphi_n\rangle\langle \psi_n,g\rangle=\langle f,g\rangle,
	$$
	and, with similar computations, $\sum_{n=0}^\infty \langle f,\psi_n\rangle\langle \varphi_n,g\rangle=\langle f,g\rangle$.
	
\end{proof}

The conclusion is therefore that, even if $(\F_\varphi,\F_\psi)$ are not necessarily made of functions in $\ltwo$, they can be used, together, to deduce a resolution (better, two resolutions) of the identity on $\E$.

\section{Examples}\label{sectexamples}

This section is devoted to the analysis of some explicit examples. In the first example $\alpha_a(x)=\alpha_b(x)=\alpha(x)$ as in the previous section, while in the second the two functions are taken to be proportional, but not equal.

\subsection{First example}\label{sectfirstexample}

Let us fix $\alpha(x)=\frac{1}{1+x^2}$. This function is always strictly positive, and produces, using (\ref{52}) and the definition of $\rho(x)$, the functions $\beta_a(x)=\rho(x)=x+\frac{x^3}{3}$ and $\beta_b(x)=\frac{-2x}{(1+x^2)^2}$. As required in Section \ref{sectAGE}, $\rho(x)\rightarrow\pm\infty$ when $x\rightarrow\pm\infty$. Also, the inverse of $\rho$ exists and can be computed explicitly looking for the only real solution of the equation $\sqrt{2}s=x+\frac{x^3}{3}$. We get
$$
x=\rho^{-1}(\sqrt{2}\,s)= \left(\frac{2}{-3\sqrt{2} s+\sqrt{2}\sqrt{2+9s^2}}\right)^{1/3}-\left(\frac{-3\sqrt{2} s+\sqrt{2}\sqrt{2+9s^2}}{2}\right)^{1/3}.
$$
The functions in (\ref{57}) turn out to be
\be
\varphi_0(x)=N_\varphi\exp\left\{-\frac{1}{2}(x+x^3/3)^2\right\}, \qquad \psi_0(x)=N_\psi\,(1+x^2).
\label{61}\en
It is clear that $\varphi_0(x)\in\ltwo$, while $\psi_0(x)$ is not square-integrable. Furthermore, see (\ref{55}), we have
$$
\pi_n(x)=\left(x+\frac{x^3}{3}\right)\pi_{n-1}(x)-\frac{1}{(1+x^2)}\,\pi_{n-1}'(x),
$$
with $\pi_0(x)=1$, and a similar expression for $\sigma_n(x)$. More explicitly we get
$$
\pi_n(x)=\sigma_n(x)=\frac{1}{\sqrt{2^n}}H_n\left(\frac{x+x^3/3}{\sqrt{2}}\right),
$$
and
\be
\varphi_n(x)=\frac{N_\varphi}{\sqrt{2^nn!}}\,H_n\left(\frac{x+x^3/3}{\sqrt{2}}\right)e^{-\frac{1}{2}(x+x^3/3)^2}, \quad \psi_n(x)=\frac{N_\psi}{\sqrt{2^nn!}}\,H_n\left(\frac{x+x^3/3}{\sqrt{2}}\right)(1+x^2),
\label{62}\en
$n\geq0$. The fact that these functions are compatible follows easily from the speed of decay of $\varphi_n(x)$, which easily contrasts (and wins) again the divergences of $\psi_n(x)$ and of the Hermite polynomials. In particular, formula (\ref{59}) shows that these functions are biorthonormal if $N_\psi N_\varphi=\frac{1}{\sqrt{2\pi}}$: $\langle\psi_m,\varphi_n\rangle=\delta_{n,m}$, $\forall n,m\geq0$. Theorem \ref{thquasibases} guarantees that $\F_\varphi$ and $\F_\psi$ are $\E$-quasi bases.

\subsection{Second example}\label{sectsecondexample}

The second example we discuss here arises out of different, but proportional, $\alpha_a(x)$ and $\alpha_b(x)$. In particular, we take $\alpha_a(x)=2\alpha_b(x)=\frac{1}{\cosh(x)}$. It is clear that, with this choice, the equality $\alpha_a(x)\alpha_b'(x)=\alpha_a'(x)\alpha_b(x)$ in (\ref{42}) is satisfied. As for $\beta_j(x)$, we take
$\beta_a(x)=\rho(x)=\int\frac{dx}{\alpha_b(x)}=2\sinh(x)$ and $\beta_b(x)=\alpha_b'(x)=\frac{-\sinh(x)}{2(\cosh(x))^2}$. We see that $\rho(x)$ is a strictly increasing function, satisfying the required asymptotic behaviour. Indeed we have $\rho(x)\rightarrow\pm\infty$ when $x\rightarrow\pm\infty$, and $\rho^{-1}$ does exist.

The vacua in (\ref{44}) are conveniently written  as
\be
\varphi_0(x)=N_\varphi\exp\left\{-(\cosh(x))^2\right\}, \qquad \psi_0(x)=2N_\psi\,\cosh(x).
\label{63}\en

It is clear that, also in this example, $\varphi_0(x)\in\ltwo$, while $\psi_0(x)$ is not square integrable. However, it is also clear that $\overline{\psi_0(x)}\,\varphi_0(x)\in\Lc^1(\mathbb{R})$. Indeed, even if $\psi_0(x)$ diverges exponentially for $|x|\rightarrow\infty$, $\varphi_0(x)$ converges to zero much faster, as $e^{-e^{|x|}}$. We will return on this aspect later.

As for the functions $\pi_n(x)$ and $\sigma_n(x)$, in this example these are no longer equal, due to the little difference between $\alpha_a(x)$ and $\alpha_b(x)$. Indeed we find, first of all, $\pi_0(x)=\sigma_0(x)=1$, and
$$
\pi_n(x)=\sinh(x)\pi_{n-1}(x)-\frac{1}{2\cosh(x)}\,\pi_{n-1}'(x),
$$
while 
$$
\sigma_n(x)=2\sinh(x)\sigma_{n-1}(x)-\frac{1}{\cosh(x)}\,\sigma_{n-1}'(x).
$$
It is now easy to prove, by induction, that these functions are (not surprisingly) again related to Hermite polynomials.
\be
\pi_n(x)=\frac{1}{2^n}\,H_n(\sinh(x)), \qquad \sigma_n(x)=H_n(\sinh(x)),
\label{64}\en
$\forall n\geq 0$. Of course, these equations imply that $\sigma_n(x)=2^n\pi_n(x)$. Hence the two functions are indeed different, but still they are proportional. Of course, we do not expect any proportionality relation survives if $\alpha_a(x)$ and $\alpha_b(x)$ are significantly different. 

Putting all together we have
\be
\varphi_n(x)=\frac{N_\varphi}{2^n\,\sqrt{n!}}\,H_n(\sinh(x))e^{-(\cosh(x))^2}, \quad \psi_n(x)=\frac{2N_\psi}{\sqrt{n!}}\,H_n(\sinh(x))\cosh(x),
\label{65}\en
$n\geq0$. A straightforward computation shows that these functions are compatible and biorthonormal if $N_\psi N_\varphi=\frac{e}{2\sqrt{\pi}}$:
\be
\langle\psi_m,\varphi_n\rangle=\delta_{n,m},
\label{66}\en
$\forall n,m\geq0$. This result is a simple consequence of the same integral between Hermite polynomials used to deduce (\ref{59}), with the change of variable $s=\sinh(x)$.

Let us now introduce a set $\E_c$ in analogy with $\E$ in (\ref{510}):
\be
\E_c=\left\{h(s)\in\ltwo:\, h_{[-]}(s):=h(\sinh^{-1}(s))\,e^{s^2/2}\in\ltwo\right\} 
\label{67}\en
This set is dense in $\ltwo$, since it contains $\D(\mathbb{R})$. Also, if $h(s)\in\E_c$, then we have
$$
h_{[+]}(s):=h(\sinh^{-1}(s))\,\frac{e^{-s^2/2}}{\sqrt{1+s^2}}\in\ltwo,
$$
as it is clear since $h_{[+]}(s)=h_{[-]}(s)\,\frac{e^{-s^2}}{\sqrt{1+s^2}}$, and using the fact that $\frac{e^{-s^2}}{\sqrt{1+s^2}}$ is bounded. To check that $(\F_\varphi,\F_\psi)$ are $\E_c$-quasi bases, we start noticing that 
\be
\langle f,\varphi_n\rangle=\frac{N_\varphi \pi^{1/4}}{\sqrt{2^n}\,e}\langle f_{[+]},e_n\rangle, \qquad \langle \psi_n,g\rangle=2\overline N_\psi \sqrt{2^n\sqrt{\pi}}\langle e_n, g_{[-]}\rangle,
\label{68}\en
for all $f(x),g(x)\in\E_c$. Here $e_n(x)$ are the eigenstates of the harmonic oscillator we have already introduced before, and the two scalar products in the right-hand sides above are well defined, since they only refer to square-integrable functions. Despite of what happens in (\ref{511}), the two results in (\ref{68}) depend on $n$ not only trough $e_n(x)$, but also because of the term $\sqrt{2^n}$. However,  these terms cancel out when we take their product, so that, using the Parceval identity for the $\{e_n(x)\}$, we get
$$
\sum_{n=0}^{\infty}\langle f,\varphi_n\rangle\langle \psi_n,g\rangle=\frac{2\overline N_\psi N_\varphi\sqrt{\pi}}{e}\sum_{n=0}^{\infty}\langle f_{[+]},e_n\rangle\langle e_n, g_{[-]}\rangle=\langle f_{[+]}, g_{[-]}\rangle,
$$ 
which is well defined, being $f_{[+]}(x), g_{[-]}(x)\in\ltwo$. With the change of variable $t=\sinh^{-1}(x)$ we find that
$$
\langle f_{[+]}, g_{[-]}\rangle=\int_{-\infty}^{\infty}\overline{f(\sinh^{-1}(x))}\,g(\sinh^{-1}(x))\frac{dx}{\sqrt{1+x^2}}=\int_{-\infty}^{\infty}\overline{f(t)}\,g(t)\,dt=\langle f, g\rangle.
$$
Similarly we prove that $\sum_{n=0}^{\infty}\langle f,\psi_n\rangle\langle \varphi_n,g\rangle=\langle f, g\rangle$. Hence $(\F_\varphi,\F_\psi)$ are $\E_c$-quasi bases.

\subsection{The related Hamiltonian operators}

In the literature on ladder operators, and on pseudo-bosonic operators in particular, an important aspect is the connection between the families $\F_\varphi$ and $\F_\psi$ with some Hamiltonian-like operator. This is because, as we have already seen in Section \ref{sectpbs}, the functions of these two sets are eigenstates of what we call here $H$ and $H^\dagger$: $H=ba$ and $H^\dagger=a^\dagger b^\dagger$. Moreover, due to the pseudo-bosonic commutation rules, $\varphi_n(x)$ is also an eigenstate of $H_{susy}=ab$, while $\psi_n(x)$ is also an eigenstate of $H_{susy}^\dagger=b^\dagger a^\dagger$, but their eigenvalues only differ by one unit by those of $H$ and $H^\dagger$. This is because $H\varphi_n=(H_{susy}+\1)\varphi_n$ and $H^\dagger\psi_n=(H_{susy}^\dagger+\1)\psi_n$, and therefore it makes not much sense to consider the SUSY partners of $H$ and $H^\dagger$, in this case. For this reason, from now on, we concentrate on $H$ and on $H^\dagger$, giving their explicit expressions in terms of the functions $\alpha_j(x)$ and $\beta_j(x)$ in (\ref{41}). In particular, using (\ref{41}) and (\ref{43}), we find the following operators
\be
H=-k_2(x)\frac{d^2}{dx^2}+k_1(x)\frac{d}{dx}+k_0(x), \qquad H^\dagger=-q_2(x)\frac{d^2}{dx^2}+q_1(x)\frac{d}{dx}+q_0(x),
\label{69}\en
where we have introduced the following functions:
\be
\left\{
\begin{array}{ll}
	k_2(x)=\alpha_a(x)\alpha_b(x),\\
	k_1(x)=\alpha_a(x) \beta_b(x)-\alpha_b(x) \beta_a(x)-2\alpha_a(x) \alpha_b'(x),\\
	k_0(x)=\beta_a(x)\beta_b(x)-\left(\beta_a(x)\alpha_b(x)\right)',\\
\end{array}
\right.
\label{610}\en
and
\be
\left\{
\begin{array}{ll}
	q_2(x)=\overline{\alpha_a(x)\alpha_b(x)},\\
	q_1(x)=\overline{\alpha_b(x) \beta_a(x)-\alpha_a(x) \beta_b(x)-2\alpha_b(x) \alpha_a'(x)},\\
	q_0(x)=\overline{\beta_a(x)\beta_b(x)-\left(\beta_b(x)\alpha_a(x)\right)'}.\\
\end{array}
\right.
\label{611}\en
Let us now show what these formulas become in the examples considered before.

\vspace{2mm}

{\bf Example 1.}

The simplest situation is when $\alpha_a(x)=\alpha_b(x)=1$, as in Section \ref{subsectionconstantalphas} with $\alpha_a=\alpha_b=1$. Hence $\beta_a(x)+\beta_b(x)=x+k$, $k$ constant in $\mathbb{R}$, and we choose $\beta_a(x)=x$ and $\beta_b(x)=k$. Hence we have $k_2(x)=q_2(x)=1$, $k_1(x)=-q_1(x)=k-x$, $k_0(x)=kx-1$ and $q_0(x)=kx$. Hence
$$
H=-\frac{d^2}{dx^2}+(k-x)\frac{d}{dx}+(kx-1), \qquad H^\dagger=-\frac{d^2}{dx^2}+(x-k)\frac{d}{dx}+kx.
$$
The eigenstates of $H$ and $H^\dagger$ are the functions $\varphi_n(x)$ and $\psi_n(x)$ deduced in Section \ref{subsectionconstantalphas}.

\vspace{2mm}

{\bf Example 2.}

Let us now deduce the expression of the operators $H$ and $H^\dagger$ for the operators considered in Section \ref{sectfirstexample}, where we have taken $\alpha(x)=\frac{1}{1+x^2}$,  $\beta_a(x)=x+\frac{x^3}{3}$ and $\beta_b(x)=\frac{-2x}{(1+x^2)^2}$. In this case, computing the functions $k_j(x)$ and $q_j(x)$ above we find
$$
H=-\frac{1}{(1+x^2)^2}\frac{d^2}{dx^2}-\frac{x(-3+7x^2+5x^4+x^6)}{3(1+x^2)^3}\frac{d}{dx}-1,
$$
and
$$
H^\dagger=-\frac{1}{(1+x^2)^2}\frac{d^2}{dx^2}+\frac{x(21+7x^2+5x^4+x^6)}{3(1+x^2)^3}\frac{d}{dx}-\frac{2(-3+18x^2+7x^4+5x^6+x^8)}{3(1+x^2)^4},
$$
whose eigenstates are given in (\ref{62}).

\vspace{2mm}

{\bf Example 3.}

The last example we want to consider here is the one discussed in Section \ref{sectsecondexample}: $\alpha_a(x)=2\alpha_b(x)=\frac{1}{\cosh(x)}$, 
$\beta_a(x)=2\sinh(x)$ and $\beta_b(x)=\frac{-\sinh(x)}{2(\cosh(x))^2}$. In this case $H$ and $H^\dagger$ are the following:
$$
H=-\frac{1}{2(\cosh(x))^2}\frac{d^2}{dx^2}+\frac{1}{2}\left((\mbox{sech}(x))^2-2\right)\tanh(x)\frac{d}{dx}-1,
$$
and
$$
H^\dagger=-\frac{1}{2(\cosh(x))^2}\frac{d^2}{dx^2}+\left(\frac{3}{2}(\mbox{sech}(x))^2+1\right)\tanh(x)\frac{d}{dx}-\frac{1}{8(\cosh(x))^4}(-9+4\cosh(2x)+\cosh(4x)),
$$
whose eigenstates are those in (\ref{65}), while the eigenvalues are, of course, the natural numbers.

\vspace{2mm}

We see that, as the last two examples clearly show, complicated Hamiltonians can be deduced using our strategy. 

\vspace{2mm}

{\bf Remark:--}
It might seem that calling Hamiltonians these operators is not entirely justified, if we imagine that an Hamiltonian should necessarily be connected to some given conservative quantum mechanical system. However, also in view of what we have discussed in the Introductio, where many physically-oriented references were cited, we still consider more than justified calling Hamiltonians the different operators $H$ and $H^\dagger$ introduced all along this section.

\section{Bi-coherent states}\label{sectBCS}

The notion of bi-coherent states has been introduced already some time ago, see \cite{bag2010} and references therein, and refined more and more in recent years. We refer to the recent monograph \cite{bagspringer} for an updated list of results and considerations on these states. In particular, in \cite{bagJPCS2021,bagspringer}, the concept of weak bicoherent states (WBCSs) has been proposed. These vectors are relevant in presence of non square-integrable eigenstates of some non self-adjoint Hamiltonian. This is exactly the situation we are discussing in this paper, where the functions $\varphi_n(x)$ and $\psi_n(x)$ do not necessarily belong to $\ltwo$, but still they are compatible and biorthonormal.

In this section, for concreteness, we concentrate on the WBCSs arising out of the functions in (\ref{65}). The extension to other situations is easy.

In analogy with what discussed in \cite{bagJPCS2021,bagspringer} we need to introduce a topology on $\E_c$: we say that a sequence $\{g_n(x)\}$ in $\E_c$ is $\tau_{\E_c}$-convergent to a certain $g(x)\in\Lc^2(\mathbb{R})$ if  $\{g_n(x)\}$ and $\{(g_n)_{[-]}(x)\}$ converge to $g(x)$ and to $g_{[-]}(x)$ respectively, in the norm $\|.\|$ of $\ltwo$. It is clear that, when this is true, $g(x)\in\E_c$. Hence, $\E_c$ is closed in $\tau_{\E_c}$. We call $\E_c'$ the set of all continuous linear functionals on $\E_c$.

It is easy to check that the following quantities, $\Phi(z)$ and $\Psi(z)$, introduced by

\be
\langle \Phi(z),g\rangle=e^{-|z|^2/2}\sum_{n\geq0}\,\frac{\overline{z}^n}{\sqrt{n!}}\langle \varphi_n,g\rangle, \label{71}\en
and \be \langle \Psi(z),g\rangle=e^{-|z|^2/2}\sum_{n\geq0}\,\frac{\overline{z}^n}{\sqrt{n!}}\langle \psi_n,g\rangle,  \label{72}\en
are well defined, for all $z\in\mathbb{C}$ and for all $g(x)\in\E_c$.

To check this we use (\ref{68}), which implies that, since $\|e_n\|=1$,
\be
|\langle g,\varphi_n\rangle|\leq \frac{|N_\varphi| \pi^{1/4}}{\sqrt{2^n}\,e}\| g_{[+]}\|, \qquad |\langle \psi_n,g\rangle|=2|N_\psi| \sqrt{2^n\sqrt{\pi}}\| g_{[-]}\|,
\label{73}\en
for all $n\geq0$. We recall that $\| g_{[+]}\|, \| g_{[-]}\|<\infty$, due to the definition of $\E_c$ and to its properties. Then it is clear that the two series in (\ref{71}) and (\ref{72}), defining $\Phi(z)$ and $\Psi(z)$, are everywhere convergent in $\mathbb{C}$, for all possible choices of $g(x)\in\E_c$. For what follows, it is now convenient to introduce two functionals on $\E_c$, $F_\Phi(z)$ and $F_\Psi(z)$, as follows:
\be
F_\Phi(z)[g]=\langle \Phi(z),g\rangle, \qquad F_\Psi(z)[g]=\langle \Psi(z),g\rangle,
\label{74}\en
for all $z\in\mathbb{C}$ and $\forall g\in\E_c$. The fact that these are linear functionals on $\E_c$ is indeed obvious. What is less clear, maybe, is the fact that they are $\tau_{\E_c}$-continuous and, because of this, define some sort of distribution. This reflects our point of view in \cite{bagJPCS2021}. What we will show here is that the same conclusions can be deduced also in the present, more general, context.

We start checking that, taken a sequence $\{g_n(x)\in\E_c\}$ which is $\tau_{\E_c}$-convergent to a certain $g(x)\in\E_c$, then $(g_n)_{[+]}(x)$ converges in $\|.\|$ to $g_{[+]}(x)$. This is because we can write
$$
(g_n)_{[+]}(x)-g_{[+]}(x)=\frac{e^{-x^2/2}}{\sqrt{1+x^2}}\left(g_n(\sinh^{-1}x)-g(\sinh^{-1}(x))\right)=\frac{e^{-x^2}}{\sqrt{1+x^2}}\left((g_n)_{[-]}(x)-g_{[-]}(x)\right).
$$
Therefore, since $\frac{e^{-x^2}}{\sqrt{1+x^2}}\leq1$, it follows that $$\|(g_n)_{[+]}(x)-g_{[+]}(x)\|\leq\|(g_n)_{[-]}(x)-g_{[-]}(x)\|\rightarrow0,$$ for $n\rightarrow\infty$, since the $\tau_{\E_c}$-convergence of $\{g_n(x)\in\E_c\}$ to $g(x)$ implies that $\{(g_n)_{[-]}(x)\}$ converges to $g_{[-]}(x)$ in the norm of $\ltwo$.

With this in mind we can check the following 

\begin{prop}
$F_\Phi(z)$ and $F_\Psi(z)$ belong to $\E_c'$.
\end{prop}

\begin{proof}
	We only have to prove that these functionals are $\tau_{\E_c}$-continuous. For that, let us consider a sequence $\{g_n(x)\in\E_c\}$ which is $\tau_{\E_c}$-convergent to $g(x)$. As we have shown this implies that $(g_n)_{[\pm]}(x)$ converges to $g_{[\pm]}(x)$ in the norm of $\ltwo$, $\|.\|$. With easy estimates we conclude that
	$$
	|F_\Phi(z)[g_n-g]|\leq e^{-|z|^2/2}\frac{|N_\varphi|\pi^{1/4}}{e}\left(\sum_{n=0}^\infty\frac{|z|^n}{\sqrt{2^n\,n!}}\right)\|(g_n)_{[+]}-g_{[+]}\|\rightarrow0,
	$$ 
	for all $z\in\mathbb{C}$. Also,
	$$
	|F_\Psi(z)[g_n-g]|\leq e^{-|z|^2/2}2|N_\psi|\pi^{1/4}\left(\sum_{n=0}^\infty\frac{(\sqrt{2}|z|)^n}{\sqrt{n!}}\right)\|(g_n)_{[-]}-g_{[-]}\|\rightarrow0,
	$$ 
	again for all $z\in\mathbb{C}$.
	
\end{proof}

\begin{prop}
	The pair $(\Phi(z),\Psi(z))$ satisfies the following properties:
	
	(i) for all $g(x)\in \D(\mathbb{R})$ we have
	\be
	\langle g,a\Phi(z)\rangle=z\langle g,\Phi(z)\rangle, \qquad \langle g,b^\dagger \Psi(z)\rangle=z\langle g,\Psi(z)\rangle,
	\label{75}\en
	for all $z\in \mathbb{C}$.
	
	(ii) We have
	\be
	\frac{1}{\pi}\int_{\mathbb{C}}\left<f,\Phi(z)\right>\left<\Psi(z),g\right>dz=
	\frac{1}{\pi}\int_{\mathbb{C}}\left<f,\Psi(z)\right>\left<\Phi(z),g\right>dz=
	\left<f,g\right>,
	\label{76}\en
	for all $f,g\in\E_c$.	
\end{prop}

\begin{proof}
	Let us check that $\langle g,a\Phi(z)\rangle=z\langle g,\Phi(z)\rangle$, for all $g(x)\in\D(\mathbb{R})$. First of all we observe that, since $a^\dagger=-\frac{d}{dx}\frac{1}{\cosh(x)}+2\sinh(x)$, $(a^\dagger g)(x)$ again belongs to $\D(\mathbb{R})$ for all $g(x)\in\D(\mathbb{R})$. Hence $g(x)\in D(a^\dagger)$, the domain of $a^\dagger$. Now, since $\D(\mathbb{R})\subseteq\E_c$, we can use (\ref{71}) to write
$$
	\langle a^\dagger g,\Phi(z)\rangle=e^{-|z|^2/2}\sum_{n\geq0}\,\frac{{z}^n}{\sqrt{n!}}\langle a^\dagger g,\varphi_n\rangle=e^{-|z|^2/2}\sum_{n\geq0}\,\frac{{z}^n}{\sqrt{n!}}\langle  g,a\varphi_n\rangle, $$
	since $\varphi_n\in D(a)$. Now, using the lowering equation $a\varphi_n=\sqrt{n}\varphi_{n-1}$, the right-hand side can be rewritten as $z\langle g,\Phi(z)\rangle$, as we had to prove. A similar proof can be repeated for the other equality in (\ref{75}).
	
	The proof of (\ref{76}) does not differ much from other analogous results, see \cite{bagspringer} for instance, and will not be repeated here.
	
\end{proof}

It is clear then why we call these {\em states} WBCSs: they are (weak) eigenstates of the pseudo-bosonic annihilation operators and they produce a resolution of the identity\footnote{In fact, they produce two resolutions of the identity, formally connected by the adjoint map.} on $\E_c$.

 It is useful to notice that, in the definition of our WBCSs, there is no mention to any displacement-like operator, as for ordinary coherent states, \cite{aagbook,didier,gazeaubook}, or as it has recently been discussed at length for bi-coherent states, \cite{bagspringer}. The analysis of their appearance, their role, and their properties is part of our future projects.

\section{Conclusions}\label{sectconcl}

This paper is still another step towards a deeper comprehension of pseudo-bosons and bi-coherent states, in particular in a situation where the role of the Hilbert space $\ltwo$ is not essential. We have argued that, in presence of pseudo-bosons, weak or not, what is really relevant is the product of the eigenstates of $H=ba$ with those of $H^\dagger$, and that this product is always in $\Lc^1(\mathbb{R})$. This is enough to consider many interesting situations, but it is still not the most general case, \cite{bag2020JPA}.

As an interesting side aspect, our general analysis gives us the possibility of finding eigenstates and eigenvalues of different, and highly non trivial, Hamiltonians. And it gives rise to WBCSs with some interesting properties which reflect, at a distributional level, the analogous properties of coherent states.

What should be investigated more, in our opinion, is the role of the displacement-like operators connected with bi-coherent states and, even more relevant, further connection with physics and with truly relevant quantum mechanical systems. This is work in progress.

From a more mathematical side, a deeper analysis of the existing relations of our approach with the PIP-spaces setting is surely interesting and worth to be considered in a close future. In particular, it would be interesting to study {\em how far we can go} with PIP-spaces, compared with the strategy proposed here. This is also work in progress.

\section*{Acknowledgements}

The author acknowledges partial financial support from Palermo University (via FFR2021 "Bagarello") and from G.N.F.M. of the INdAM.

\end{document}